\documentclass[12pt]{article}
\usepackage[utf8]{inputenc}
\usepackage{amsmath,amsthm,tikz,tcolorbox}
\usepackage{mathtools}
\usepackage{amssymb}
\usepackage{amsopn}
\usepackage{float}
\usepackage{bbm}
\usepackage{graphicx}
\usepackage{subcaption}
\usepackage{float}
\usepackage{color}
\usepackage{hyperref,url}
\usepackage[margin=22mm]{geometry}
\tikzstyle{level 1}=[level distance=3.5cm, sibling distance=3.5cm]
\tikzstyle{level 2}=[level distance=3.5cm, sibling distance=2cm]
\tikzstyle{bag} = [text width=9em, text centered]
\tikzstyle{end} = [circle, minimum width=3pt,fill, inner sep=0pt]

\usepackage{pgfplots}
\usetikzlibrary{arrows.meta}

\pgfplotsset{compat=newest,
    width=6cm,
    height=3cm,
    scale only axis=true,
    max space between ticks=25pt,
    try min ticks=5,
    every axis/.style={
        axis y line=left,
        axis x line=bottom,
        axis line style={thick,->,>=latex, shorten >=-.4cm}
    },
    every axis plot/.append style={thick},
    tick style={black, thick}
}
\tikzset{
    semithick/.style={line width=0.8pt},
}
\usepgfplotslibrary{groupplots}
\usepgfplotslibrary{dateplot}

\newcommand{\ba}{\begin{eqnarray}}
\newcommand{\ea}{\end{eqnarray}}
\newcommand{\Life}{\Lambda}
\newtheorem{theorem}{Theorem}[section]

\newtheorem{lemma}[theorem]{Lemma}
\newtheorem{assumption}[theorem]{Assumption}
\newtheorem{definition}[theorem]{Definition}
\newtheorem{proposition}[theorem]{Proposition}
\newtheorem{remark}[theorem]{Remark}
\newtheorem{example}[theorem]{Example}
\newcommand{\cadlag}{c\`adl\`ag }

\title{Model-free  Analysis of Dynamic Trading Strategies}
\author{Anna ANANOVA\thanks{Mathematical Institute, University of Oxford, UK. {\bf Email}: \url{ananova.math@gmail.com}. } \,\, Rama CONT\thanks{Mathematical Institute, University of Oxford, UK. {\bf Email}: \url{Rama.Cont@maths.ox.ac.uk}.} \,\, and Renyuan XU\thanks{Department of Finance and Risk Engineering, New York University, US. {\bf Email}: \url{rx2364@nyu.edu}. R.X. is partially supported
by the NSF CAREER award DMS-2339240 and a JP Morgan Faculty Research Award.}}

\date{\today}

\begin{document}

\maketitle
\begin{abstract}
We introduce a model-free approach for  analyzing the risk and return for a broad class of dynamic trading strategies, including pairs trading, mean-reversion trading  and other statistical arbitrage strategies, in terms of {\it excursions} of a trading signal away from a reference level. Our results are derived in a pathwise setting, without any probabilistic assumptions.

We introduce the notion of {\it $\delta$-excursion}, defined as a path which deviates by $\delta$ from a reference level  before returning to this level. We show that every continuous path has a unique decomposition into {$\delta$-excursions}. This decomposition is useful for the scenario analysis of dynamic trading strategies, leading to simple expressions for the number of trades, realized profit, maximum loss, and drawdown. We show that the high-frequency limit of mean-reversion strategies may be described in terms of the ($p-$th order)
local time of the signal. In particular, our results yield a financial interpretation of the local time of an irregular path. Finally, we describe a non-parametric scenario simulation method for generating paths whose excursion properties match those observed in empirical data.
\end{abstract}
\section{Introduction}
A broad class of trading strategies may be described in terms of the relation between the {\it market price} $P_t$ of an asset --a stock, bond, commodity, a spread between two such assets, or a basket of assets-- and a {\it reference} level $A_t$, which may refer to an assessment of the portfolio's fundamental value by an analyst, or a forecast of the portfolio's value based on `technical' indicators, such as moving average estimators used in pairs trading \cite{pairs} or 'technical indicators' used in  statistical arbitrage strategies \cite{alexander1999,anca2005,avellaneda2010,hogan2004}.
The deviation $ S = P - A $ of the market price from the reference value then represents a {\it trading signal}.
If $S$ falls below some negative threshold $-\delta<0$, this represents a buying opportunity, while if $S$
exceeds a positive threshold $\delta>0$, this represents an opportunity for entering a short position.
A wide range of trading strategies -- pairs trading \cite{gatev2006,pairs}, mean-reversion strategies \cite{avellaneda2010,leung2015}, statistical arbitrage strategies based on cointegration \cite{alexander1999}, index arbitrage \cite{anca2005} and other statistical arbitrage strategies \cite{avellaneda2010,hogan2004}-- fall under this description. The reference level $A_t$ is computed differently in each of these examples, but once the signal $S=P-A$ is constructed all  these strategies follow the description given above.

Regardless of how the reference value $A_t$ is arrived at, e.g. using fundamental valuation principles, or statistical forecasts, this leads to similar features across all such trading strategies: a long position is entered when the signal $S$ crosses $-\delta$ and held until $S$ crosses $0$; similarly, a short position is entered when $S$ crosses $\delta$ and held until $S$ crosses zero.
The holding periods of positions thus coincide with  {\it excursions} of the signal $S$ above (or below) certain levels.

This remark has %far-reaching 
interesting implications: it implies that the risk and return profile of such  trading strategies may be described in terms of the properties of  excursions  of the process $S$. For example,  the profit of such a strategy is linked to the number of the excursions alluded to above, while the magnitude of drawdown risk  may be linked to the height of the excursions.

%Analytical and probabilistic properties of excursions of stochastic processes have been studied in detail, starting with P. L\'evy \cite{levy1948} for Brownian motion, K. Ito \cite{ito1972} for Markov processes.Ito's seminal contribution was to note that the excursions of a Markov process from a level may be viewed as a collection of independent random variables in a (infinite-dimensional) space of excursions. This infinite-dimensional viewpoint has proved extremely fruitful for the theoretical study of stochastic processes \cite{blumenthal2012,dynkin1971,pitman2003,rogers1989,rogers1981,williams1974}.

%We develop this idea show that the concept of excursion is very relevant for the model-free analysis  of  dynamic trading strategies, and illustrate this through the example of pairs trading strategies.

\paragraph{Contribution}
We introduce a novel approach for  analyzing the risk and return for a broad class of dynamic trading strategies, including pairs trading, mean-reversion trading  and other statistical arbitrage strategies
 in terms of {\it excursions} of the signal away from a reference level. 
 Our results are formulated in a pathwise, model-free setting, without any probabilistic assumptions on price dynamics.  

We introduce the notion of {\it $\delta$-excursion}, defined as a path which deviates by $\delta$ from a reference level  before returning to this level. We show that every continuous path has a unique decomposition into {$\delta$-excursions}. This decomposition is shown to be useful for the scenario analysis of dynamic trading strategies, leading to simple expressions for the number of trades, realized profit, maximum loss, and drawdown. 

We show that the high-frequency limit of mean-reversion strategies, which corresponds to the case where the trasding threshold $\delta$ decreases to zero, 
may be described in terms of the ($p-$th order)
local time of the signal. In particular, our results yield a financial interpretation of the local time of an irregular path, as the high-frequency  limit of the profit of a mean-reversion trading strategy.

Finally, we describe a non-parametric scenario simulation method for generating paths whose excursion properties match those observed in empirical data.

%\paragraph{Relation with previous literature}
The construction and empirical performance of pairs trading \cite{pairs,gatev2006} and `mean-reversion' trading strategies \cite{avellaneda2010,leung2015}  considered in this paper have been  studied  by Avellaneda \& Lee \cite{avellaneda2010},
Gatev et al. \cite{gatev2006} and others \cite{pairs,hogan2004}. Leung and Li \cite{leung2015} study mean-reversion strategies from the perspective of optimal control, in the setting of the Ornstein-Uhlenbeck model.
The connection between statistical arbitrage and cointegration has been discussed by many authors, including Alexander  \cite{alexander1999} and Alexander \& Dimitriu \cite{anca2005}.
Our approach provides a different perspective on these results through the angle of excursion theory and explains the common features observed across the variety of strategies considered in these studies.

Excursion theory has also been applied in mathematical finance, for the pricing of certain path-dependent options involving barrier crossings of a price process, such as Parisian options \cite{chesney1997,dassios2010}, barrier options \cite{pistorius2007} or ``occupation time derivatives'' \cite{cai2010}.
These studies focus on analytical results for special models such as Brownian motion \cite{chesney1997} or certain L\'evy processes \cite{cai2010,pistorius2007}.

A   related topic is the modeling of {\it drawdown risk} for trading strategies \cite{grossman1993}. The literature on this topic has focused on the analytical study  of drawdown risk and optimal investment under drawdown constraints in specific models. Zhang \cite{zhang2015} uses excursion theory for one-dimensional diffusion models to derive formulas for drawdown risk of static portfolios. On the other hand empirical studies of drawdown risk indicate that commonly used stochastic models do not correctly quantify drawdown risk even for passive index portfolios \cite{johansen2002}, suggesting that better, more flexible   models are needed.  

%, a topic we also address in this paper.

\paragraph{Outline}
We
propose a {\it model-free} framework for the  analysis of such dynamic trading strategies, based on a description in terms of  {\it excursions} of the underlying trading signal.  

We start in Section \ref{sec:strategy} by describing how properties of a large class of trading strategies may be expressed in terms of excursions of a {\it trading signal} away from zero. We then introduce in Section \ref{sec:pathwise} the notion of {\it $\delta$-excursion}, defined as a path which deviates by $\delta$ from a reference level  before returning to this level. We show that every continuous path has a unique decomposition into such { $\delta$-excursions}, which turns out to be useful for the scenario analysis of dynamic trading strategies, leading to simple expressions for the number of trades, realized profit, maximum loss and drawdown (Section \ref{sec.scenario}).  
In the case of irregular paths which possess a local time, we describe in Section \ref{sec.localtime}  the relation between $\delta$-excursions and  local time at zero of the path. 

%In we   

In Section \ref{sec:modelfreesimulation} we   propose  a non-parametric scenario simulation method for generating paths whose excursions match those observed in a data set.

\section{Mean-reversion strategies} 
\label{sec:strategy}

\subsection{Trading signals}	\label{sec.signals}
Many  trading strategies are based on the assumption that the market price $P_t$ of a reference asset reverts to a `target value' or forecast $A$, %assumed to be constant over a certain trading horizon\footnote{The framework can be easily generalized to accommodate a time-dependent target that changes values at the end of each trade.}, 
although it may deviate from it in the short term. The examples below illustrate the generality of this concept.
%Consider a  portfolio whose price, denoted $(P_t, t\geq 0)$, is assumed to be a continuous function of time: $P\in C^0([0,\infty),\mathbb{R})$. We will refer to this portfolio as the {\it asset}. 

		\begin{example}[Value trading]{\em An investor who believes that the price of the asset will eventually revert to a `fundamental' value $A>0$ will choose to buy the asset when $P_t$ drops below $A$ and short the asset when $P_t$ exceeds $A$. This `fundamental' value can be a book value or a valuation by a financial analyst. The deviation $S_t=P_t-A$ from the fundamental value then plays the role of trading signal.}
	\end{example}
		\begin{example}[Pairs trading]{\em  Pairs trading is a relative-value trading strategy which looks for pairs of assets whose prices $P^1,P^2$ are cointegrated \cite{alexander1999}, i.e. there exists a stationary combination $P_t=P^1_t-w P^2_t$.  $w$ is typically estimated using regression techniques \cite{gatev2006}. If $A$ is the stationary mean of $P_t$ then the deviation $S_t=P^1_t-w P^2_t -A$ is expected to revert to zero and is used as a trading signal. In practice, this means $A$ is estimated as time  average of past values \cite{pairs}. }
  \end{example}
		
		\begin{example}[Mean-reversion strategies]{\em Many {\it statistical arbitrage strategies} \cite{avellaneda2010,hogan2004} are based on identifying combinations of assets (portfolios) whose market price follows a stationary, mean-reverting process \cite{anca2005,avellaneda2010}, using methods such as index tracking or cointegration \cite{anca2005}.
		
		The market price $P_t=\sum w_i P^i_t$ of such a stationary combination is then expected to revert to its  mean $A$, which may be estimated using for instance an average over the previous trading period, leading to the trading signal   $S_t= \sum w_i P^i_t - A$ which is expected to revert to zero.
	}\end{example}

	%	\begin{example}[Delta-hedging]{\em A very different type of dynamic strategy involving level-crossings arises in hedging strategies of derivatives portfolios. Given a derivatives portfolio with value $V(t,X_t)$ depending on an underlying asset $X_t,$ the widely used 'delta-hedging' strategy aims to maintain a position $\phi_t$ in the underlying asset such that the overall sensitivity or 'delta' of the portfolio is zero: $\Delta(t,X_t)=\phi_t+\partial_x V(t,X_t)\simeq 0$. In presence of transaction costs it is not feasible to do this continuously and trades are executed when $|\Delta(t,X_t)|\geq \delta$	where $\delta$ represents a limit on the portfolio sensitivity. Here $S_t=\Delta(t,X_t)$ acts as a signal for triggering transactions.	}\end{example}
	
	%	\begin{example}[Statistical arbitrage strategies]{\em More generally, {\it statistical arbitrage strategies} \cite{avellaneda2010,hogan2004} use a statistical model to forecast the price of a reference portfolio over a time horizon $\Delta$ and take a long (resp. short) position in the portfolio depending on whether the current market price $P_t$ is lower (resp. higher)	than the forecast $A_t$. If the forecast $A_t$ is  unbiased then the trading signal 	$ S_t=  P_t-A_t$ has mean zero and fluctuates around zero.	}\end{example}
		These strategies, while distinct in their design, share a common feature: they are based on the assumption that a trading signal $ S_t= P_t - A$, defined as the deviation of the market price $P_t$ of a reference asset from a target value $A$,  reverts to zero over some time horizon. This assumption implies that if $S_t<0$ (resp. $S_t>0$) one should take a long (resp. short) position in the portfolio $P$.

{In the presence of transaction costs, such transactions will be entered only if the amplitude of the signal  reaches some threshold   $\delta$ larger than the transaction cost per trade:}
	 \begin{itemize}
	     \item[(i)] Enter a long position in the reference portfolio   when $S_t$ drops below $-\delta$; unwind the long position when $S_t$ crosses zero;
	     \item[(ii)] Enter a short position in the portfolio when $S_t$ exceeds $\delta$; unwind the short position when $S_t$ crosses zero. \end{itemize}
	 Such a strategy may be implemented through limit orders placed at the appropriate price levels, resulting in transactions when the market price $P_t$ crosses these  levels.

We now describe the associated trading strategies and their properties in more detail.
%	Let us assume we have an average value $a$, and we believe that the price reverts to this value, and let $\delta>0$ be the threshold of the price move from its average after which we decide to trade. Unless specified otherwise, we shall assume $S_0=0$.
	 \subsection{Representation of mean-reversion %trading 
  strategies in terms of excursions}\label{sec:trades_and_excurson}
	Regardless of how  the  signal $S$ is constructed, the trading strategies in the above examples share some common features, which may be described in terms of the {\it level crossings} of the signal $S$. 
	
	We define the following level crossing times of $S$ {(with $S_0=0$)}: we set $\tau^+_0=0$, $\theta^+_0=0$ and 
	\ba  \forall i\geq 1,\quad \tau^+_i=\inf\{ t>\theta^+_{i-1}, S_t\geq \delta \}& & \theta^+_i=\inf\{ t>\tau_{i}^+, S_t\leq 0 \}.\label{eq.tau} \ea
%	The stopping times  $\theta_{i-1}^+\leq \tau^+_i\leq \theta_i^+\leq ..$  mark the successive excursions of $S$ from $0$ to $\delta$ and back to zero.
	The intervals $(\tau^+_i, \theta^+_i),\, (\theta^+_i, \tau^+_{i+1})$ are the down-crossing and up-crossing intervals of the interval $[0, \delta].$
Each interval $[\theta^+_i, \theta^+_{i+1}]$, corresponds to an {\it excursion} of $S$ from $0$ to $\delta$ and back to zero. 

	It is readily observed that the intervals {$ (\theta^+_i, \tau^+_{i+1}), (\tau^+_{i+1}, \theta^+_{i+1})$, for $i \ge 0$, }form a partition of $[0,\infty)$ and, if the path is continuous,
	they are all non-empty. 
 %\footnote{ In the case of \cadlag paths a finite number of these intervals could be empty, i.e. $\tau^+_i =\theta^+_i$, if the process jumps across the interval $(0, \delta)$. In this paper we focus on the case of continuous trajectories.}
	One can also define similar quantities for downward  excursions: 
	\ba \tau^-_0= \theta^-_{0}=0,\quad{\rm and}\quad \forall i\geq 1,\quad \tau^-_i=\inf\{ t>\theta^-_{i-1}, S_t\leq -\delta \}& & \theta^-_i=\inf\{ t>\tau^-_{i}, S_t\geq 0 \}. \ea
	
	A mathematical description of the trading strategies described in Section \ref{sec.signals} can now be given in terms of the level crossing times defined above:
	\begin{itemize}
	    \item  buy the reference portfolio when the trading signal drops below $-\delta$, sell when it returns to $0$:
	    \ba\phi^- =  \sum_{k\geq 1} 1_{[\tau^-_k, \theta^-_k)}.\label{eq.phi-}\ea
	    \item short the reference portfolio when the signal exceeds $\delta$,  unwind the position when it reaches $0$:  \ba\phi^+ = - \sum_{k\geq 1} 1_{[\tau^+_k, \theta^+_k)}.\label{eq.phi+}\ea
	\end{itemize}
	We refer to $\phi^+, \phi^-$ as one-sided strategies.

		 Combining the two strategies we obtain what is usually called a  'mean-reversion strategy' or 'convergence trade' based on the trading signal $S$:
	\ba\phi^0(t)=\phi^+(t)+\phi^-(t)=  \sum_{k\geq 1} 1_{[\tau^-_k, \theta^-_k)}- \sum_{k\geq 1} 1_{[\tau^+_k, \theta^+_k)}.\label{eq.2sided}\ea
		One may also consider a position size which depends on the level $S$.
		For example, \eqref{eq.2sided} has unbounded exposure to price movements and in most cases portfolios are subject to position limits or exposure limits ('stop loss'). A maximum exposure  limit of $M$ on short positions in \eqref{eq.phi+} leads to unwinding the position if $S$ reaches $\delta+M$ during the holding period:
			\ba\phi^+_M(t)= - \sum_{k\geq 1} 1_{[\tau^+_k, \theta^+_k \wedge \kappa_k )}\qquad \kappa_k=\inf\{ t>\tau^+_k, {S_t} \geq \delta+M\}.\label{eq.Klimit}\ea
We assume the target price is revised at a  lower frequency, outside of holding periods, so that over the horizon $[0,T]$ of of the analysis it is held constant \footnote{The framework can be  generalized without difficulty to accommodate a time-dependent target $A_t$ which is updated outside of holding periods.}:
 \begin{assumption}[Trading signal]\label{ass.signal}
     {\em The trading signal $S$ has the form $S_t=P_t-A$ where
     \begin{itemize}
         \item $P_t$ is the market price of a (basket of) traded  asset(s), assumed to be continuous.
         \item $A$ is a 'target' (forecast) value, assumed to be constant over the horizon $[0,T]$.
     \end{itemize}}
 \end{assumption}
 
 Given that  $A$ is a constant, we have $\int_0^t \phi  dS =  \int_0^t \phi dP $ for $0\leq t\leq T$. %$$  \int_0^. \phi(t)  dS_t =  \int_0^. dP_t\qquad \int_0^.  \phi(t) 1_{[\tau^-_k, \theta^-_k )}(t) dS_t = \int_0^.  \phi(t) 1_{[\tau^-_k, \theta^-_k)}(t) dP_t $$
 In addition to the position in the risky asset(s), each portfolio has a cash component, which is adjusted to reflect the gains and losses from trading, so that the strategy is self-financing. 
 Denoting by $V_t(\phi)$ the sum of the cash holdings and the market value of a position $\phi$ in the risky asset, we have
\ba V_t(\phi)=V_0(\phi)+\int_0^t \phi(u-)dS_u. \label{eq.self-financing}\ea
%{\color{blue}in which we assume that $A$ is a forecast value that remains constant.} % Every time we unwind the position, we can adjust the forecast value. Therefore at the jump time of $A$, we never hold any position $\phi = 0$. The trading frequency is much higher, typically occurring on a minute or hourly basis, compared to the updates of $A$, which usually happen on a daily or weekly basis. This discrepancy is a defining characteristic of mean-reversion strategies.
 { Note that as the strategies $\phi$ considered above are piece-wise constant, no further assumption on $S$ is required to define the integral in \eqref{eq.self-financing}.}

	As the sets $\cup_{k\geq 1}[\tau^-_k, \theta^-_k]$ and $\cup_{k\geq 1}[\tau^+_k, \theta^+_k]$ are disjoint we may study the properties of $\phi^+,\phi^-$ separately. In the following sections we will focus on $\phi^+$, but it is clear that properties of $\phi^-$ are analogously obtained by replacing $S$ by $-S$.

Let us now examine further   the properties of the one-sided strategy \eqref{eq.phi+}.
Each transaction cycle $[\theta^+_{k-1},\theta^+_k]$ is decomposed into a {\it waiting period} $[\theta^+_{k-1},\tau^+_k]$ followed by a {\it holding period} $[\tau^+_k, \theta^+_k]$.
The strategy {$\phi^+$} generates a profit of $\delta$ over each transaction cycle, leading to a portfolio value
	 \ba  V_t(\phi^+)=V_0(\phi^+)+\delta\  D^\delta_t(S)+   S_{t\wedge \tau^+_{ D^\delta_t(S)+1}}-S_{t\wedge \theta^+_{ D^\delta_t(S)+1}}, \quad {\rm where}\quad D^\delta_t(S)= \sum_{i\geq 1} 1_{\theta_i^+\leq t } % < \infty %&\qquad{\rm and}\qquad U^\delta_t(S)= \sum_{i\geq 1} 1_{ \tau_i^+ \leq t}.
\ea	
 represents the number of transactions in $[0,t]$ {of $\phi^+$}. The first term $\delta\  D^\delta_t(S)$ represents the {\it realized profit} while the second term corresponds to the market value of the current position.
If the path of $S$ wanders high above $\delta$ then the portfolio can incur a large market loss. 
It is therefore clear that the gains and losses of the trading strategy $\phi^+$ are linked to the frequency, duration and amplitude of positive {\it excursions} of $S$ which exceed the level $\delta$. Similarly, one can readily observe that the gains and losses of $\phi^-$ are linked to the frequency, duration and height of  negative {\it excursions} of $S$ which reach $-\delta$.
In the following sections, we build on this insight and study in more  detail the structure of such excursions in order to model the risk and return profile of such portfolios.
\section{Pathwise results and scenario analysis}
\label{sec:pathwise}
 	
 	\subsection{Excursions and $\delta-$excursions}
 Let ${\cal E}=C^0([0,\infty), \mathbb{R})$ be the space of continuous functions equipped with the Borel measurable structure induced by the uniform norm and 
  ${\cal E}_{0}=\{f\in {\cal E}, f(0)=0\}$. Denote, for $f\in {\cal E},$
	\ba T^{x}(f)=\inf\{t> 0, f(t)= x\},\qquad T_t^{x}(f)=\inf\{u>t, f(u)= x\}.\label{eq.hittingtime}\ea
	Let $\delta\in \mathbb{R}$. {We define ``an excursion from $0$ to $\delta$'' as a path which starts from zero, reaches $\delta$ in a
finite time, and is stopped when it reaches $\delta$:}%We will call an {\it excursion} from $0$ to $\delta$ a path  which starts from zero, reaches $\delta$ in a finite time, and stops when it reaches $\delta$:
	\ba{\cal E}_{0,\delta}=\{ f:C^0([0,\infty)\to \mathbb{R}{)},\, f(0)=0,\, T^\delta(f)< \infty ;\, \forall t\geq T^\delta(f), f(t) = \delta\,\}.\label{eq.Ed}\ea
	Note that by this definition an excursion from $0$ to $\delta$ is stopped at the first time it reaches $\delta$.
	In particular, ${\cal E}_{0,0}$ is the space of {\it excursions} from $0$ to $0$.
	
	Define the {\it concatenation} at $T>0$ of two paths $u,v\in {\cal E}$ as the element
	\ba (u \mathop{\oplus}_{T} v)(t):=  u(t)\ 1_{[0,T)}+v\left(  t-T\right) 1_{[T,\infty)}. \ea
	Note that if $ u\in {\cal E}_{0,a}, v\in {\cal E}_{a,0}$ then for $T\geq T^a(u),$ $u \mathop{\oplus}_{T} v\in {{\cal E}_{0,0}}$.
	
 We define a {\it $\delta-$excursion} as an excursion  from $0$ to $\delta$, followed by an excursion from $\delta$ back to $0$:
   \begin{definition}[$\delta$-excursion]
{\em A $\delta-$excursion is a path $f\in {\cal E}$ such that
\ba  \exists  (u,v)\in {\cal E}_{0,\delta}\times  {\cal E}_{\delta,0},\,\, f=u \mathop{\oplus}_{T^{\delta}(u)} v, \qquad{\rm i.e.}\quad f(t)=u(t)\ 1_{[0,T^{\delta}(u))}+v\left(  t-T^{\delta}(u)\right) 1_{[T^{\delta}(u),\infty)} \label{eq:delta_excursions}\ea
    The  decomposition \eqref{eq:delta_excursions} is then unique and we denote $\Life(f)=T^\delta(u)+ T^{0}(v)$ the {\em duration} of $f$.
    
We denote by ${\cal U}_{\delta}$ the set of $\delta-$excursions.
%\ba\label{eq.U_delta} {\cal U}_{\delta} = {\cal E}_{0,\delta} \oplus {\cal E}_{\delta,0}.\ea
The map $f\in {\cal U}_{\delta} \mapsto (u,v,\Lambda(f))\in {\cal E}_{0,\delta}\times  {\cal E}_{\delta,0}\times [0,\infty)$ is measurable.}
  \label{def:delta_excursions} \end{definition}

Examples of  $\delta-$excursions are
 excursions from $0$ to $0$ which reach $\delta$:
\ba\label{eq:Gammadelta}
\Gamma_\delta=\Big\{ f\in {\cal E}_{0,0}\colon\, \max(f)\geq \delta \Big\}=  {\cal U}_\delta\cap {\cal E}_{0,0}.\,\ea 
 The inclusion $\Gamma_\delta\subset {\cal U}_\delta$ is strict, as a typical  $\delta-$excursion  may reach zero (infinitely) many times before reaching $\delta$ and we may have $\Life(f) > T^0(f)$ for $f\in {\cal U}_{\delta}$. 
In particular ${\cal U}_{\delta}$ is {\it not} a subset of ${\cal E}_{0,0}$.  
		However, each path in ${\cal U}_\delta$ contains {\it exactly} one excursion of type $\Gamma_\delta$:
	\begin{lemma}[Last exit decomposition of $\delta$-excursions] Any $\delta$-excursion $f\in {\cal U}_\delta$ has a unique decomposition into a path from $0$ to $0$ which does not reach $\delta$ {\it followed} by an excursion $\gamma\in \Gamma_\delta$ from $0$ to $0$ which reaches $\delta$:
	\ba \forall f\in {\cal U}_\delta,\quad \exists!(T,g,\gamma)\in [0,\infty)\times {\cal E}_{0}\times \Gamma_\delta,\quad f=g\mathop{\oplus}_{T}\gamma  \quad{\rm with}\quad g(0)=g(T)=0,\quad \max(g)<\delta.\label{eq.Udecomposition}\ea\label{lemma.Ud}\end{lemma}
	\begin{proof} Consider a $\delta$-excursion $f\in {\cal U}_\delta$. Then $f$ has a decomposition \eqref{eq:delta_excursions} for some $(u,v)\in {\cal E}_{0,\delta}\times {\cal E}_{\delta,0}$ and $f(t)=0$ for $t\geq \Lambda(f)=T^\delta(u)+T^{0}(v)$. Now define $T$ as the last zero of $u$ before $T^{\delta}(u)$:
	$$ T=\sup\{ t< T^{\delta}(u),\ u(t)=0\}.$$
	Then by continuity of $u$, $T< T^{\delta}(u)$ and therefore $\max\{ u(t), 0\leq t\leq T\}< \delta$. Setting $g=f 1_{[0,T]}$ and {$\gamma(t)=f(t+T)$}, it is readily verified that $\gamma\in \Gamma_\delta $ and $g$ satisfy the required conditions. 

 {Uniqueness of the triple $(T, g, \gamma)$ follows from the fact that, from the decomposition $f=g\mathop{\oplus}_{T}\gamma$, we can identify $\gamma$ as the unique excursion of $f$ starting at $T$ and ending at $\Lambda(f)$, and $g=f\large|_{[0,T]}$.}
	\end{proof}
	\begin{figure}[H]
	\centering
	\input{projection_U_Pi.pgf}
	\caption{Example of last exit decomposition of a $\delta$-excursion $f\in \mathcal{U}_{\delta}$ with $\delta=1$: $f=g\mathop{\oplus}_{T}\gamma$ where  $g$ is in purple and $\gamma\in \Gamma_\delta$ (in orange) is the last excursion. }
	\end{figure}
	\subsection{Decomposition of a path into $\delta-$excursions}
		The following proposition gives the decomposition of any path    starting  from zero into a sequence of excursions from $0$ to $\delta$ and back to $0$:
	\begin{proposition}
	    Let $S\in C^0([0,\infty),\mathbb{R})$ with $S_0=0$. 
	Define the level crossing times %$(\tau^+_i,\theta_i^+)_{i\geq 1}$  as in \eqref{eq.tau}:
	$$ \theta_0^+=0, \quad\tau_0^+=0,\qquad  \tau_i^+=T^\delta_{\theta_{i-1}^+}(S), \qquad\theta_i^+ =T^{0}_{\tau_i^+}(S).$$
  \ba{\rm Then}\qquad \forall t\geq 0,\quad & D^\delta_t(S)= \sum_{i\geq 1} 1_{\theta_i^+\leq t } < \infty &\qquad{\rm and}\qquad %U^\delta_t(S)= \sum_{i\geq 1} 1_{ \tau_i^+ \leq t}
		\label{eq.upcrossings}\\ 
		\forall t\geq 0, \qquad S_t&=  \sum_{i=1}^{D^\delta_t(S)+1} \Big[ u_i\left( t-\theta_{i-1}^+\right) 1_{[\theta_{i-1}^+,\tau_{i}^+)}  &+ v_{i}\left( t-\tau_{i}^+\right)1_{[\tau_{i}^+,\theta_{i}^+)} \Big], \label{eq.decomposition}\ea
		where $u_i \in {\cal E}_{0,\delta}$ and $v_i \in {\cal E}_{\delta,0}$.\label{prop.decomposition}
	\end{proposition} 
	\begin{proof} 
	To prove the first assertion, we first note that $S$ is continuous,  thus uniformly continuous on $[0,T]$ for any $T>0$.
	If $D^\delta_t=\infty$ for some $t>0,$ then the set $\{k\in \mathbb{N}, \quad \theta^+_{k}\leq t\}$ is infinite. Since by construction the intervals $(\tau^+_{k}, \theta^+_{k})$ are disjoint, we have
	\[
	    \sum_{\{k, \theta^+_{k}\leq t\}} |\theta^+_{k} -\tau^+_{k}| \leq t <\infty,\quad {\rm so}\quad \inf_{\theta^+_{k}\leq t}|\theta^+_{k} -\tau^+_{k}|= 0\quad{\rm while}\qquad |S_{\theta^+_{k}} -S_{\tau^+_{k}}| =  \delta
	\]
which	contradicts the uniform continuity of $S$ on $[0, t]$.  Therefore $D^\delta_t<\infty$  for all $t\geq 0$.
Starting from:
		\[
			S_t = \sum_{i=1}^{D^\delta_t(S)+1} \Big[ (S_{t\wedge \tau^+_{i}} - S_{t\wedge \theta^+_{i-1}} )  + (S_{t\wedge \theta^+_{i}} - S_{t\wedge \tau^+_{i}} )\Big].
		\]
		Note that
		\[
		    (S_{t\wedge \tau^+_{i}} - S_{t\wedge \theta^+_{i-1}} )  + (S_{t\wedge \theta^+_{i}} - S_{t\wedge \tau^+_{i}} )=
		    (S_{t\wedge \tau^+_{i}} - S_{t\wedge \theta^+_{i-1}} )1_{[\theta_{i-1}^+,\tau_{i}^+)}  + (\delta+S_{t\wedge \theta^+_{i}} - S_{t\wedge \tau^+_{i}} )1_{[\tau_{i}^+,\theta_{i}^+)}.
		\]
	{Setting, for $i\geq 1$,
		\begin{eqnarray*}
		    u_i(t) &:=&  S_{(t+\theta^+_{i-1})\wedge \tau^+_{i}} - S_{(t+\theta^+_{i-1})\wedge \theta^+_{i-1}} = S_{(t+\theta^+_{i-1})\wedge \tau^+_{i}},\\ 
			v_i(t) &:=& \delta + S_{(t+\tau^+_i)\wedge \theta^+_{i}} - S_{(t+\tau^+_i)\wedge \tau^+_{i}} = S_{(t+\tau^+_i)\wedge \theta^+_{i}}.
		\end{eqnarray*}	
		we obtain the desired decomposition. }
	\end{proof}
	The above results translate into a (measurable) decomposition of any continuous path into $\delta$-excursions:
		\begin{proposition}[Decomposition of a path into $\delta-$excursions]\label{prop.pathdecomposition}
		%$D^\delta_t(S), U^\delta_t(S)$ are respectively the number of downcrossings and upcrossings of the interval $(a, \delta)$ in $[0,t]$
	Let $\delta >0$ and $S\in C^0([0,\infty),\mathbb{R})$ with $S_0=0$, and define $D^\delta_t(S)$ as in \eqref{eq.upcrossings}.
	\begin{enumerate}
	    \item[(i)] If $\sup_{t\geq 0} D^\delta_t(S)=\infty$   
	      there exists a unique sequence $(e_k)_{k\geq 1}$ of $\delta-excursions $ $e_k\in {\cal U}_\delta$ such that
	\ba\forall t\geq 0,\quad S_t=   \sum_{k\geq 1} e_k\left( (t-\theta^+_{k-1})_+\right)\qquad{\rm where}\qquad \theta^+_0=0,\quad \theta^+_k= \sum_{i=1}^k \Lambda(e_i).\label{eq.dexcursions}\ea
	    \item[(ii)] If $d=\sup_{t\geq 0} D^\delta_t(S)<\infty$ then there exist  $(e_1, \ldots,e_d)\in ({\cal U}_\delta)^d $ and $e_{d+1}\in { {\cal E}_0 \backslash \mathcal{U}_\delta}$ such that
	    \ba S_t=   \sum_{k=1}^{d+1} e_k\left( (t-\theta^+_{k-1})_+\right)\qquad{\rm where}\qquad \theta^+_0=0,\quad \theta^+_k= \sum_{i=1}^k \Lambda(e_i).\label{eq.finiteexcursions}\ea
	%\item[(iii)]   $\sup(S)<\delta $: in this case $D^\delta_t(S)=0$ for all $t\geq  0$ and ..
 \end{enumerate}
	In all cases the map $S\mapsto (e_k)_{k= 1, \ldots, d+1}$ is measurable.
	\end{proposition}
	The case (i) corresponds to the 'recurrent' case where the path crosses zero and $\delta$ infinitely many times on $[0,\infty)$.
		\begin{proof}	
	Set $d=\sup_{t\geq 0}D^{\delta}_t(S)\in \mathbb{N}\cup \{\infty\}$.  
		Define $(\theta^+_k, k\geq 1)$ as in \eqref{eq.tau}. 
		For $k<d$, set {$e_k(t)=S_{t+\theta^+_{k-1}}\,1_{[0,\theta^+_k-\theta^+_{k-1})}(t)$}. Then it is easily verified, from the definition \eqref{eq.tau} of $\theta_k^+$, that $e_k\in {\cal U}_\delta $ and $\Lambda(e_k)=\theta^+_k-\theta^+_{k-1}$. Measurability of the map $S\mapsto (e_k)_{k\geq 1}$ follows from the measurability of the hitting times and the shift operator.
		To show uniqueness, we note that  \eqref{eq.dexcursions} implies that $e_k(\cdot - \theta^+_{k-1})=S\large|_{[\theta^+_{k-1},\theta^+_k)}$ so it is sufficient to show uniqueness of the sequence $(\theta^+_k)_{k\geq 0}$. As $D^\delta_t(S)<\infty$ for each $t>0$, the countable set $\{t>0, \Delta D^\delta_t\neq 0\}$ is discrete and has a unique increasing ordering, which is given by $(\theta^+_k)_{k\geq 0}$.
	\end{proof}
	\begin{remark}{\em The above results decompose the path into one-sided $\delta$-excursions i.e. with $\delta>0$. One can immediately obtain a similar decomposition for $\delta<0$ by applying the above result to the path $-S$. To obtain a decomposition in terms of {\it two-sided} $\delta$-excursions, one can iterate these two results: first decompose $S$ into $\delta$-excursions, then decompose each $\delta$-excursion into $(-\delta)-$excursions. One may further show that the resulting decomposition is independent of the order of these two operations.}
	\end{remark}
\subsection{Scenario analysis for mean-reversion strategies}\label{sec.scenario}

The {\it drawdown} \cite{grossman1993} of a  portfolio $\phi$ whose value at time $t$ is $V_t(\phi)$  is defined as 
$$ \Delta(t)= M_t(\phi)- V_t(\phi)\qquad{\rm where}\quad M_t(\phi)=\mathop{\max}_{s\in[0,t]} V_{s}(\phi) $$
is the running maximum. 
 
The decomposition of the path into $\delta-$excursion  given in Proposition \ref{prop.pathdecomposition} leads to simple expressions for the portfolio value, the maximum loss and the drawdown of the strategy:
	\begin{proposition} 
		 Along a path $S$ with decomposition \eqref{eq.decomposition}, 	
		 \begin{enumerate}
		     \item[(i)] the gain 	$V_t(\phi^+)-V_0(\phi^+)= \int_0^t \phi^+ dS$ of the portfolio is given by
		     \ba V_t(\phi^+)-V_0(\phi^+)= \delta \times D_t^\delta(S) + 1_{[\tau^+_{D_t^{\delta}+1},\theta^+_{D_t^{\delta}+1}]}\left(\delta -v_{D^{\delta}_t(S)+1}(t-\tau^+_{D_t^{\delta}+1}) \right). \ea
		     \item[(ii)] the worst loss during $[0,t]$ is given by
		\ba\max_{s\in [0,t]}\left( V_0(\phi^+)-V_s(\phi^+)\ \right)= \mathop{\max}_{k=0,\ldots, D^{\delta}_t(S)} \{   \max_{[0, (t - \tau^+_{k+1})_+]}\left(v_{k+1}- (k+1)\delta \right)\}.\label{eq:worst_loss} \ea
		\item[(iii)] the drawdown of $\phi^+$ is given by
	\ba
 \Delta(t)&=&	\mathop{\max}_{k=0,\ldots, D^{\delta}_t(S)} \{   \max_{[0, (t - \tau^+_{k+1})_+]}\left( (k+1)\delta-v_{k+1} \right)\}\nonumber\\
 &&-\delta \times D_t^\delta(S) - 1_{[\tau^+_{D_t^{\delta}+1},\theta^+_{D_t^{\delta}+1}]}(t)\left(\delta -v_{D^{\delta}_t(S)+1}(t-\tau^+_{D_t^{\delta}+1}) \right).\label{eq:drawdown}
	\ea
		%\ba \Delta(t)=\left( v_{D^\delta_t(S)+1}( (t-\tau^+_{D^\delta_t(S)+1})_+) \right)_+ +\mathop{\max}_{\left[0,\left(t-\tau^+_{D^\delta_t(S)+1}\right)\right]}\left(v_{D^\delta_t(S)+1}\right)_{-}
	%	\ea
		 \end{enumerate} 
		
	\end{proposition}
		\begin{proof}
			By definition of the portfolio $\phi^+$, we have
			\[
			V_t(\phi^+)= 	V_0(\phi^+)- \sum_{i=1}^{D^\delta_t(S)+1} \int 1_{[\tau ^+_{i}, \theta^+_{i})} dS = 
				\sum_{i=1}^{D^\delta_t(S)} (S_{\tau^+_i} - S_{\theta^+_i}) +   (S_{t\wedge \tau^+_{D^{\delta}_t(S)+1}} -S_t ),
			\]
			thus
			$	V_t(\phi^+)-	V_0(\phi^+)= \delta \times D_t^\delta(S) + S_{t\wedge \tau^+_{D^{\delta}_t(S)+1} } -  S_t.$ By the definition of $v_i$ this can be rewritten as
			$$	V_t(\phi^+)-	V_0(\phi^+)= \delta \  D_t^\delta(S) + 1_{[\tau^+_{D_t^{\delta}+1},\theta^+_{D_t^{\delta}+1}]}\left(\delta -v_{D^{\delta}_t(S)+1}(t-\tau^+_{D_t^{\delta}+1}) \right), $$ which then implies \eqref{eq:worst_loss}.
{Furthermore,
$$
M(t) = V_0(\phi^+) + {\max}_{k=0,\ldots, D^{\delta}_t(S)} \left\{   \max_{[0, (t - \tau^+_{k+1})_+]}\left( (k+1)\delta-v_{k+1} \right)\right\},
$$
which together with the previous formula yields \eqref{eq:drawdown}.}
			%$\min_{[0,t]} V_t(\phi^+)-	V_0(\phi^+)= \mathop{\min}_{i=0,\ldots, D^{\delta}_t(S)} \{\delta  i - \max(v_{i+1})\}.$
		\end{proof}
		\begin{figure}[h]
		\centering
	\input{figures/price.pgf}
	\input{figures/value.pgf}
	\input{figures/drawdown.pgf}
	\caption{Decomposition of a path into $\delta$-excursions with $\delta=0.5$. Top Left:  decomposition into $u_i$ (blue) and $v_i$ (red). Top Right: Value of the portfolio $\phi^+$. Bottom: Drawdown $\Delta(t)$.}
	\end{figure}

\section{High-frequency asymptotics}
\label{sec.HF}
{For {\it irregular} price paths, the frequency of level crossings for levels $\delta $ close to zero is connected with the concept of {\it local time} of the path \cite{geman1980}. 
We now explore this connection in a pathwise framework and show that it leads to a financial interpretation for the mathematical concept of local time.}

 \subsection{Irregular price paths}\label{sec.localtime}
	
	Intuitively, decreasing the value of the threshold $\delta$  increases the frequency of level crossings and leads to more transactions but with a lower profit  per transaction. For irregular price paths, such as sample paths of stochastic processes, the frequency of level crossing may go to infinity as  $\delta$ decreases to zero (while the profit $\delta$ per transaction  goes to zero), so the  behaviour of the profit in this limit is not clear.
    The exact behaviour of the trading strategy as $\delta\to 0$ is determined by the {\it local time} of the path at $0$, which measures the time spent by the path in a neighbourhood of zero \cite{geman1980}.
	
%	admits a finite quadratic variation $[S]_t$ (along a given fixed sequence of partitions $\left(\pi_n\right)_{n\geq 1}$ of $[0, \infty)$). 

Let $S\in C^0([0,\infty),\mathbb{R})$ and $T>0$. The occupation measure of   $S$ is defined by
	\ba 
	\gamma_T(A):=  \int_0^T 1_A(S_t)\, dt, \qquad\forall A\in \mathcal{B}(\mathbb{R}).
	\ea
	We will say that the path $S$ admits a local time $l_T(S,x)$ if the measure $\gamma_T$ is absolutely continuous with respect to Lebesgue measure on $\mathbb{R}$, in which case we denote
	\[
	l_T(S,x) := \frac{d \gamma_T}{dx} = \lim_{\varepsilon \to 0} \frac{1}{2\varepsilon} \int_0^T 1_{[x -\varepsilon , x+\varepsilon]}(S_t)\, dt.
	\]
	The occupation density is characterized by the {\it occupation time formula}:
	\[
	\int_0^T h(S_t)\, dt = \int_{-\infty}^{+\infty} h(x)\,l_T(S, x)\,dx,\, \forall h \in C^0(\mathbb{R},\mathbb{R}).
	\]
		Intuitively, the local time $l_T(S,x)$ represents the time  $S$ spends at level $x$ during $[0,T]$. We will be interested in particular in the local time at $0$, which we denote 
	$\ell_T(S)=l_T(S,0)$.
	
	The map	$T\mapsto \ell_T(S)$ is increasing, which allows to
	%The notion local time is essential in the theory of Ito excursion, which are excursion of the path $S$ in the set ${\cal E}_{0,0}$. 
	define its right-continuous inverse, the {\it inverse local time} at zero:
	\ba \forall l>0,\qquad \tau_{l}=\inf\{ t>0, \,\,\ell_{t}(S)> l \}. \label{eq.tau_l}\ea
	We note that $l\mapsto \tau_l$ is an increasing \cadlag function of the variable $l.$  
	The occupation density at zero $\ell_t(S)$ increases on the set $\{ t, S_t=0\}$ and is constant along any excursion from $0$, so the discontinuities of $\tau$ correspond to excursions of $S$, and
	  jump intervals of $\tau$ correspond to the complement of the set where $S$ visits  $0$:
	\[
		\mathop{\cup}_{l>0} (\tau_{l-}, \tau_l)=  \{t\geq 0, \, S_t \neq 0\}.
	\]
	Thus the value $l$ of local time  along an excursion may be used as a natural index for labeling excursions of $S$: the excursion at  local time level $l$ is given by
\ba\label{eq:S_excursion} e_l(t, S) =
\begin{cases}
S_{\left(\tau_{l-}+t\right)} 1_{\left(t \leq \tau_{l}-\tau_{l{-}}\right)},&\text { if }\quad \tau_{l}(\omega)-\tau_{l-}(\omega)>0\\ 
	\dag &\text{ if }\quad \tau_{l}(\omega)=\tau_{l-}(\omega).
\end{cases}
\ea
Points of continuity of $\tau_l$, i.e. points at which $\tau_{l-}=\tau_l$ correspond to `infinitesimal excursions' which may arise if the path  has non-zero local time at $0$;  we associate such excursions with a `cemetery' state $e_l(S) = \dag $.
This defines an excursion process $e\colon \mathbb{R}^+ \to \overline{{\cal E}_{0,0}}={\cal E}_{0,0}\cup \{\dag\}.$ For a given set $\Gamma\subset {\cal E}_{0,0}$, we can define the counting process, which counts excursions of  $S$ from $0$ which lie in $\Gamma$, up to  local time $l$ :
\ba
	N_l(\Gamma):= \sum_{\lambda\leq l} 1_{\Gamma}(e_{\lambda}).\label{eq.N}
\ea
Note that in general $N_l(\Gamma)$ can be infinite. We now establish an important connection between this {\it excursion point process} $N$ and the decomposition into $\delta$-excursions given by Proposition \ref{prop.pathdecomposition}. 
Recall the set $\Gamma_\delta$ of excursions from $0$ to $0$ which reach a level $\delta$:
$$\Gamma_\delta=\Big\{ f\in {\cal E}_{0,0}\colon\, \max(f)\geq \delta \Big\}.\,$$
\begin{proposition}\label{prop.NGammad}
	Let $S\in C^0([0,\infty),\mathbb{R})$ be a path with $S_0   =0$ which admits an occupation density $\ell_t(S)$ at zero, with inverse $\tau$ is given by \eqref{eq.tau_l}. For $\delta>0$ let $D^\delta_t(S)$ be the number of $\delta-$excursions of $S$ on $[0,t]$, defined as in {\eqref{eq.upcrossings}}. Then 
	\[
	\forall \delta>0,\quad \forall t >0,\qquad	N_t(\Gamma_{\delta})<+\infty,\, \quad{\rm and}
	\]
	\ba
	\forall t>0,\quad	D^{\delta}_t(S) = N_{\ell_t(S)}(\Gamma_{\delta})\quad \text{\rm and } \quad\forall l>0, \quad D^{\delta}_{\tau_l}(S) = N_{l}(\Gamma_{\delta}) .
\label{eq.identity}	\ea
\end{proposition}
%\eqref{eq.identity} provides the link between  $\Gamma_{\delta}$ and $U_{\delta}$: $	D^{\delta}_t(S) $ counts the number of $\delta$-excursions in $U_{\delta}$ along the path $S$.  
% This section is pathwise results, no need to mention ITO yet.
%On the other hand, $\Gamma_{\delta} \subset {\cal E}_{0}$ and we can apply Ito's representation to  $N_{\ell_t(S)}(\Gamma_{\delta})$. \eqref{eq.identity} bridges $U_{\delta}$ with existing results on excursions. See more discussions in Section \ref{sec:markov}.

\begin{proof}
	The condition $N_l(\Gamma_{\delta})<+\infty$ is a consequence of the continuity of $S$. Recall the level crossing times defined in \eqref{eq.tau}. We will now establish a one-to-one correspondence between excursions $e_l\in\Gamma_{\delta}$ and intervals $(\theta^+_{i-1}, \theta^+_i).$ As in Lemma \ref{lemma.Ud}, define $\hat{\theta}_i^+$ $i\geq 1$ the 'last exit' from zero in the $i$-th $\delta-$excursion:
	\[
		\hat{\theta}^+_{i} := \sup \{t<\tau^+_{i}\colon S_t = 0\}.
	\]
To show that the two sets of intervals $\{(\hat{\theta}^+_{i}, \theta^+_{i})\}_{i\geq 1}$ and $\{(\tau_{l-}, \tau_l)\}_{e_{l}\in \Gamma_{\delta}}$ coincide, we prove the following two claims:
	\begin{itemize}
		\item For each $i\geq 1$ there exists a unique $l_i\geq 0, $ such that $(\hat{\theta}^+_{i}, \theta^+_{i}) =  (\tau_{l_i-}, \tau_{l_i})$.
		
			Indeed, it is easy to see that on the interval $(\hat{\theta}^+_{i}, \theta^+_{i})$,  $S_t  >0$. Furthermore, $$\hat{e}_i(t):=S_{t+ \hat{\theta}^+_{i}} 1_{[0,\, \theta^+_{i}-\hat{\theta}^+_{i}]}\in \Gamma_{\delta},$$ 
since $\hat{e}_i(\tau^+_i-\hat{\theta}^+_i)\geq \delta.$
		In particular  $(\hat{\theta}^+_{i}, \theta^+_{i})$ is an interval of   $\{t>0, \, S_t \neq 0\}$, thus there exists unique $l_i$ such that
		$(\hat{\theta}^+_{i}, \theta^+_{i}) =  (\tau_{l_i-}, \tau_{l_i})$ and $e_{l_i}\in \Gamma_{\delta}$. 
		
		\item  Conversely, for every $l\geq 0$ such that $e_l\in \Gamma_\delta$, there exists a unique index $i(l)\geq 1$ such that $(\tau_{l-}, \tau_{l}) = (\hat{\theta}^+_{i(l)}, \theta^+_{i(l)})$.
		
	Take the largest $i=i(l)\geq 1$ such that $\theta^+_{i-1}\leq \tau_{l-}$ then $\tau_{l-}<\theta_{i}^+$. Since on $(\hat{\theta}^+_{i}, \theta^+_{i})$,  $S_t  >0$, while $S_{\tau_{l-}}  =0,$ we get that $\hat{\theta}^+_{i}\geq \tau_{l-}$. The condition $e_l\in \Gamma_\delta$ implies that $S$ reaches the level $\delta$ in $ (\tau_{l-}, \tau_{l})$,  by definition $\tau^+_i>\hat{\theta}^+_{i}$ is the first such time  after $\theta_{i-1}$, hence $\tau^+_i\in (\tau_{l-}, \tau_{l})$. Since the intervals $(\tau_{l-}, \tau_{l})$ and $(\hat{\theta}^+_{i}, \theta^+_{i})$ intersect, we conclude from the first claim that $(\tau_{l-}, \tau_{l}) = (\hat{\theta}^+_{i(l)}, \theta^+_{i(l)})$ (we also use the fact that the intervals $\{(\tau_{l-}, \tau_{l})\}_{l\geq 0}$ are disjoint).

	\end{itemize}

The correspondence between $\theta_{i}^+,\, i\geq 1$ and $\tau_l,\, e_l\in \Gamma_{\delta}$, yields the result:
\[
	D^\delta_t(S)= \sum_{i\geq 1} 1_{\theta_i^+\leq t}= \sum_{i\geq 1} 1_{\theta^+_{i(l)}\leq t}  =  \sum_{\tau_l\leq t} 1_{\Gamma_\delta}(e_{l}) = \sum_{l\leq l_t(S)} 1_{\Gamma_\delta}(e_{l}) = N_{l_t(S)}(\Gamma_\delta).
\]
\end{proof}
%The identity \eqref{eq.identity} implies in particular that $$

%	\cite{davis2018}, define a discrete local time $L^{\delta \mathbb{Z}}_t(S)$, related to a sequence of Lebesgue partitions associated to a grid $\delta \mathbb{Z}$.
%	
%	
%	They observe that the up-crossing down-crossing numbers $U_t(S, (a, \delta))$, $D_t(S, (a, \delta))$ of the interval $(a, \delta)$ (the latter in our notation is $N^\delta_t(S)$) are related to the local time $L^{\delta \mathbb{Z}}_t(a)$ as follows:
%	$$
%	L_{t}^{\pi_{\delta} Z}(a) / 2=D_t(S, (a, \delta))(\delta-a)+D_t(S, (a, \delta)) a+1_{\left[S_{t_{j}}, S_{t}\right)}(a)\left|S_{t}-a\right|.
%	$$
%	Where $t_j$ is the last partition point in the interval $[0, t].$ In particular, they conclude that
%	\[
%	\left|L_{t}^{\pi_{\delta} Z}(S) / 2 - \delta\, D_t(S, (a, \delta))  \right| \leq 2 \delta,
%	\]
%	thus in our notation
%	\[
%	\left|L_{t}^{\pi_{\delta} Z}(S) / 2 - \delta\, N_{t}^{\delta}(S)  \right| \leq 2 \delta.
%	\]

\subsection{Behaviour of level-crossings as $\delta\to 0$}
\label{sec:asymototics}
The behavior of the above quantities as $\delta \to 0$ is determined by the `roughness' of the path. 
When $\delta$ is small, we account for the fact that trading takes place only at prices that are integer multiples of a 'tick', i.e. only at times when $S$ takes such values.

Let $\delta_n=2^{-n}$ and  introduce the  partition $\pi_{n}$ defined by the hitting  times of the grid $\delta_n \mathbb{Z}$:
	\ba\label{eq.lebesguepartition}
	t_{0}^{n}:=0, \quad t_{k+1}^{n}:=\inf \left\{t \geq t_k^{n}\colon S_{t} \in  \delta_n \mathbb{Z}\backslash \{S_{t_{k}^{n}} \} \right\}.
	\ea
Then $\sup_{\pi_n}|S_{t^n_{k+1}}-S_{t^n_k}|\to 0$ as $n\to\infty$.
We denote $\pi = (\pi_n)_{n\geq 1}$ {with $\pi_n:=\{t_k^n: k \ge 0\}$}. %$|\pi_n|=\sup_{k}|t^n_{k+1}-t^n_k|\to 0$ as $n\to\infty$.
Following \cite{contperkowski}, we will say that $S\in C^0([0,T],\mathbb{R})$ has $p$-th order variation  along $\pi$ if there exists $[S]^p_\pi\in C^0([0,T],\mathbb{R}_+)$ such that
$$\sum_{\pi_n}|S_{t^n_{k+1}\wedge t}-S_{t^n_k\wedge t}|^p\quad \mathop{\to}^{n \rightarrow \infty} \quad[S]_\pi^p(t).$$
{Note that this limit is uniform in $t\in[0,T]$, as pointed out by \cite{contperkowski}.}
%Denote by $V_p(\pi)$ the set  of continuous paths satisfying this property.We shall assume $S$ satisfies this property and choose \ba p=\inf \{q\geq 1, S\in V_q(\pi), [S]^q(T)>0 \} \ea
The smallest  $p\geq 1$ for which $[S]^p_\pi \neq 0$ then gives an index of `roughness' for $S$ along $\pi$. For example for Brownian paths $p=2$ while for fractional Brownian motion with Hurst exponent $H$, $p=1/H$ \cite{biagini2008}.

For paths with non-zero $p$-th variation, the number of down-crossings for levels close to zero is related to a slightly different notion of local time, defined in terms of a {\it weighted} occupation measure, weighted by  the p-th order variation \cite{contperkowski,contjin2024}:
\begin{definition}[Local time of order $p$ \cite{contperkowski}]
	 	Let $p\geq 1$ and $q\geq 1$. A continuous path $S\in C^0([0, T], \mathbb{R})$ has ($L^q$-)local time of order $p$ along a  sequence of partitions $\pi = (\pi_n)_{n\geq 1}$ of  $[0, T]$ if, for any $t\in [0,T]$, the sequence of functions
	 	$$
	 L^{\pi_n, p}_t(S, .): x\in \mathbb{R}\mapsto	L^{\pi_n, p}_t(S, x):= \sum_{t^n_{j} \in \pi_n} 1_{\left[S_{t^n_{j} \wedge t}, S_{t^n_{j+1} \wedge t}\right)}(x)\left|S_{t^n_{j+1} \wedge t}-x\right|^{p-1} %\quad\in L^q([0,T])
	 	$$
	 	 	converges, {as $n$ goes to infinity}, in $L^q(\mathbb{R})$ to a limit $L^{\pi,\, p}_t(S,\cdot)\in L^q(\mathbb{R})$  and the map $t\in [0,T]\mapsto L^{\pi,\, p}_t(S,x)\in L^q(\mathbb{R})$ is weakly continuous. We call $L^{\pi,\, p}(S,x)$
	 	the local time of order $p$ of $S$ at level $x$.
	 \end{definition}
	 $L^{\pi,\, p}_t(S,x)$ measures the rate at which the path $S$ accumulates p-th order variation around level $x$. Note that the local time of order $p$ is non-zero only if $S$ has non-zero p-th order variation along $\pi$ i.e. $[S]^p_\pi>0$.
	If the convergence is uniform in  $(t, x)\in [0,T]\times \mathbb{R}$, and the mapping $(x, t)\mapsto L^{\pi, p}_t(S, x)$ is continuous we call it the continuous local time	of $S$ \cite{kim2019local}.

	Note that  in the case of $p=2$ the definitions  in \cite{bertoin1987}  and \cite{contperkowski, kim2019local} differ by a factor of $2$; here we use the latter notation. In the case $p=2$ we will omit the index $p$ in the notation; $L^{\pi}:= L^{\pi,\, 2}$.
	
{The quantity $D^\delta_t(S)= \sum_{i\geq 1} 1_{\theta_i^+\leq t ,}$ is the number of down-crossings of the interval $[0, \delta_n]$, as defined in \eqref{eq.upcrossings}.
 It is also equal to the total number of $\delta$ excursions and the number of transactions of strategy $\phi^+$ defined in \eqref{eq.phi+}. 
Similarly we define, as  in \cite{contperkowski}, the number of up-crossings of  $[0, \delta_n]$
$$U^\delta_t(S)= \sum_{i\geq 1} 1_{\tau_i^+\leq t }.$$
Following a reasoning similar to \cite[Lemma 3.4]{contperkowski},
 for any $x\in[0,\delta_n]$ the summands in the formula for $L^{\pi_n, p}_t(S, x)$ are either $|x|^{p-1}$ for each down crossing or $|\delta_n -x|^{p-1}$ for up crossing of the interval, and (additionally)  a final term $|x-S_t| = O(\delta_n^{p-1})$ in the case $S_t\in[0,\delta_n]$. Thus
  	$$ 
	L^{\pi_{n}, p}_t(S,x)=D^{{\delta_n}}_t(S)|x|^{p-1}+U^{{\delta_n}}_t(S) |\delta_n -x|^{p-1}+ O(\delta_n^{p-1}), \forall x\in[0,\delta_n].
	$$ }
 % Following the arguments in %\cite{bertoin1987,elkaroui1978} and
 % \cite[Lemma 3.4]{contperkowski}, one can establish a relation between the  down-crossings and up-crossings $D^{\delta_n}_t(S), U^{\delta_n}_t(S)$ of the interval $ [0,\delta_n]$: 
 % for $x\in [0,\delta_n],$
 % 	$$ 
	% L^{\pi_{n}, p}_t(S,x)=D^{{\delta_n}}_t(S)|x|^{p-1}+U^{{\delta_n}}_t(S) |\delta_n -x|^{p-1}+ O(\delta_n^{p-1}).
	% $$ 
 %for $x\in [0,\delta_n],$
%	$$ 
%	L^{\pi_{n}, p}_t(S)=D^{\delta_n}_t(S)|x|^{p-1}+U^{\delta_n}_t(S) |\delta_n -x|^{p-1}+ O(\delta_n^{p-1}).
%	$$ 
	Since the numbers  $D^{{\delta_n}}_t(S), U^{{\delta_n}}_t(S)$ can differ at most by one, we  obtain that  
	$$
		L^{\pi_n}_t(S) = L^{\pi_n}_t(S, 0)=D^{\delta_n}_t(S) \delta_n^{p-1}+ O(\delta_n^{p-1}).
	$$ 
	 If $S$ has a continuous local time $L^{\pi, p}(S, \cdot)$ along the sequence of Lebesgue partitions $\pi$, we conclude from above that
	\[
		\lim_{n\to \infty}|\delta_n|^{p-1} D^{\delta_{n}}_t(S) = L^{\pi, p}_t(S).
	\]
	%\begin{assumption}	\end{assumption}
 The following proposition summarizes the behavior of the number of level crossings $D^\delta_t(S)$ (representing the number of trades) and the realized profit  $\delta D^\delta_t(S)$ as $\delta$ decreases to zero:
	\begin{proposition} Let $\delta_n=2^{-n}$ and $p\geq 1$. Assume $S\in C^0([0,T],\mathbb{R})$ has a strictly positive local time $L^{\pi,p}_t>0$ of order $p$ at zero along the sequence of partitions $(\pi_n)_{n\geq 1}$ defined by \eqref{eq.lebesguepartition}. Then for any $t\in (0,T]$,
		\begin{enumerate}
		    \item[(i)] if $1\leq  p<2$ then \ \ $\delta_n\  D^{\delta_n}_t(S)  \mathop{\to} 0$ as ${n\to \infty}$. 
	    \item[(ii)] if $p>2$ then $$\delta_n\  D^{\delta_n}_t \mathop{\longrightarrow}^{n\to \infty} \infty,\qquad {\rm and}\quad D^{\delta_n}_t(S) \mathop{\sim}^{\delta_n\to 0}\  %\delta_n^{1-p} 
	    \frac{L^{\pi,p}_t(S)}{\delta_n^{p-1}}.$$
	    \item[(iii)] if $p=2$ then $$\delta_n\  D^{\delta_n}_t \mathop{\longrightarrow}^{n\to \infty}  L^{\pi,2}_t(S),\qquad {\rm i.e.}\quad D^{\delta_n}_t(S) \mathop{\sim}^{\delta_n\to 0}\  \frac{L^{\pi,2}_t(S)}{\delta_n}.$$
		\end{enumerate}\label{prop.localtime}
	\end{proposition}
In particular, when $p>2$ the threshold $\delta$ should be chosen as small as possible, while for $p\leq 2$ there is an optimal threshold $\delta^*(S)>0$ which maximizes the realized profit $\delta D^{\delta}_T(S)$.

{When transaction costs are $c>0$ per share, $\delta\geq c$ so the above asymptotic regime will be relevant only when the transaction cost per share is small i.e. for liquid instruments.
}

	    The assumptions of Proposition \ref{prop.localtime} are satisfied by typical sample paths of many classes of stochastic processes.
	    Typical paths of semimartingales correspond to (iii), while paths of 'rough' processes such as Fractional Brownian motion with Hurst exponent $H<1/2$ correspond to (ii). Below we provide a few examples.
	   
     \begin{example}[Continuous semimartingales]{\em
	 Let $S=M+A$ where $M$ is a continuous martingale and $A$ is a continuous process with bounded variation $\int_0^T|dA_t|$ on $[0,T]$. Denote by $[S]=[M]$ the quadratic variation process of $S$.
	 Then $S$ admits a local time of order $p=2$, which corresponds $1/2$ of the 'semimartingale local time' of $S$ at $0$:
		\[
		L^{\pi,2}_t(S)  =  \lim_{\varepsilon \to 0} \frac{1}{4\varepsilon} \int_0^T 1_{[ -\varepsilon , \varepsilon]}(S_t)\, d[S]_t.
	\]
{For any continuous semimartingale, the following (one-sided) version holds:
\[
    L^{\pi,2}_t(S)  =  \lim_{\varepsilon \to 0} \frac{1}{2\varepsilon} \int_0^T 1_{[ 0 , \varepsilon]}(S_t)\, d[S]_t.
\]
}
Furthermore, as shown by El Karoui \cite{elkaroui1978},  if for some $q\geq 1$ we have $$ \mathbb{E}\left[ [M]_T^{q/2}+  \left(\int_0^T|dA_t|\right)^q\right]< \infty,  $$ then  $t\mapsto \delta\  D_{t}^{\delta }$ is uniformly approximated   in $\mathbb{L}^q$ by $L^{\pi,2}_t(S)$ as $\delta \to 0$:
	$$
	\mathbb{E}\left[ \sup _{0\leq t \leq T}\left|\delta\  D_{t}^{\delta }(S)- L^{\pi,2}_t(S)\right|^q\right]\quad \mathop{\to}^{\delta \rightarrow 0}\quad 0,
	$$
}
	    \end{example}
	    
	     \begin{example}[Fractional Brownian motion]{\rm Let $B^H$ be a fractional Brownian motion  with Hurst parameter $H\in (0, 1)$.
	    {The  
almost-sure convergence of the $p$-th variation  of $B^H$  along the the sequence of  partitions $\pi$  defined in \eqref{eq.lebesguepartition} for $p=1/H$
has been shown by Das et al. \cite{das2023level}, who also discuss   excursions of this process.}   
        Thus $B^H$  almost surely has a continuous local time $L^{\pi, 1/H}(B^H, \cdot)$ of order $p=1/H$ along the Lebesgue partitions \eqref{eq.lebesguepartition}, and
	$$L_t^{\pi, 1/H}(B^H, x)= c_H \ell_t(B^H, x),$$ where $\ell_t(B^H, \cdot)$ is the occupation time density of $B^H$ and $c_H$ a constant (see also \cite{kim2019local}). 
    Denoting $\ell_t(B^H)=\ell_t(B^H, 0)$, this implies
	\[
	    \lim_{n\to \infty}|\delta_n|^{\frac{1-H}{H}} D^{\delta_n}_t(B^H) = \ell_t(B^H)\underbrace{\mathbb{E}\left[|B^H_1|^{\frac{1}{H}}\right]}_{c_H},\quad {\rm so}\quad 
D^{\delta_n}_t(B^H)\mathop{\sim}^{\delta_n\to 0}\frac{c_H\  \ell_t(B^H)}{|\delta_n|^{\frac{1-H}{H}}}.	\]  

} 
	    \end{example}

\begin{example}[Fractional Ornstein-Uhlenbeck process]
The     fractional Ornstein-Uhlenbeck process \cite{cheridito2003} is a Gaussian process solution of a Langevin equation driven by a fractional Brownian motion:
\begin{eqnarray}
	 dS_t= -\lambda S_t dt + \gamma dB^H_t, \label{eq.FOU}
	 \end{eqnarray}
	 where $B^H$ is a fractional Brownian motion with Hurst exponent $H$.
     The solution of this equation is given by
     $$ S_t= e^{-\lambda t}\left(S_0+ \gamma \int_0^t e^{\lambda u}dB^H_u\right).$$
	$S$ is an ergodic process which exhibits long-range dependence \cite{cheridito2003}.
   If $E[S_0]=0,$ the process is recurrent at zero.
     
	Figure \ref{fig.smalldelta}  shows, as a function of the threshold $\delta$,  the number of $\delta-$excursions estimated from values of $S_t$  on a discrete grid of $N=28,800$ points (which corresponds to the number of seconds in one trading day). 
The empirical estimator closely follows the asymptotics described in Proposition \ref{prop.localtime}, suggesting that this asymptotic regime is indeed a relevant description of excursions at such frequencies.

\begin{figure}[H]
  \centering
  \begin{subfigure}[b]{0.32\textwidth}
    \centering
    \includegraphics[width=\textwidth]{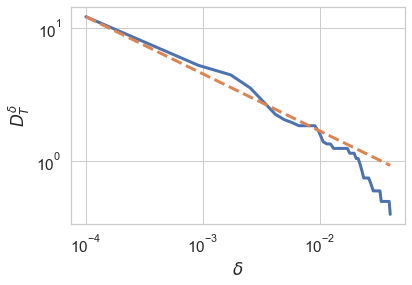}
    \caption{ $H=0.7$.}
  \end{subfigure}
  \begin{subfigure}[b]{0.32\textwidth}
    \centering
    \includegraphics[width=\textwidth]{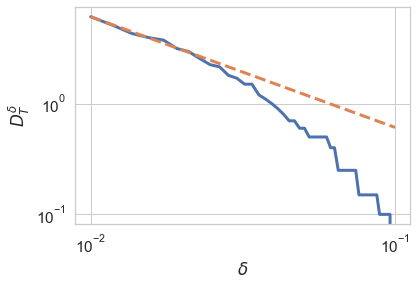}
    \caption{ $H=0.5$.}
  \end{subfigure}
    \begin{subfigure}[b]{0.32\textwidth}
    \centering
    \includegraphics[width=\textwidth]{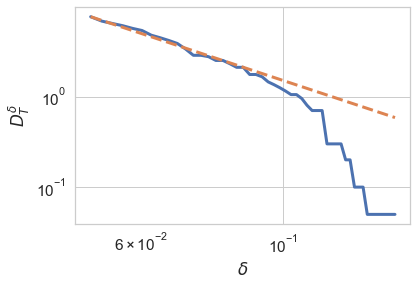}
    \caption{$H=0.3$.}
  \end{subfigure}
  \caption{Behavior of $D^{\delta}_T$ when $\delta\to 0$ for a fractional Ornstein-Uhlenbeck process \eqref{eq.FOU} with $\lambda=5$ and $\gamma=0.1$. Dotted line: asymptotic behavior $\delta^{\frac{H-1}{H}}$ described in Proposition \ref{prop.localtime}.}\label{fig.smalldelta}
\end{figure}
Many of the properties of this example may be extended to other stochastic differential equations driven by fractional Brownian motion \cite{hairer2005}.
\end{example}

%{\bf Add discussion on choice of $\delta$} scaling of $\tilde{\ell}_{t}(S)$ under the transformation $\sigma(t)^2=d[S]/dt \to \lambda \sigma(t)^2   $

\section{Application to pairs trading}\label{sec:behavior}

 Pairs trading \cite{pairs,gatev2006} is a trading strategy based on identifying a stationary linear combination of two stock prices  and using this linear combination as a trading signal for generating buy/sell transactions in the pair. In most applications the signal is then modeled as an AR(1)/Ornstein-Uhlenbeck process \cite{leung2015,pairs}. However, such model-based methods  usually rely on  strong (and often unrealistic) model assumptions hence may suffer from potential financial losses.

 We illustrate the limitation of  model-based methods with the empirical distributions of  two pairs trading signals, both of which are recognized to have co-movements.

 The first pairs trading signal is constructed with  CocaCola (KO) and PepsiCola (PEP) and the second one is constructed with two ETFs ProShares Short S\&P500 (SH) and ProShares UltraShort S\&P500 (SDS). %Both pairs are  recognized to have co-movements -- KO and PEP are stocks of two companies mainly producing cola and SH and SDS are both EFTs for shorting S\&P500. 
 In both examples, as we shall see below, we observe poor performance of the fitted Ornstein-Uhlenbeck process by comparing  the empirical distributions with the real pairs trading signals.

We use second-by-second NYSE price records of KO, PEP, SH, SDS shares  during trading hours  09:30AM-4:00PM  for the period 07/01/2013 - 07/01/2020 to construct pair-trading signals. Denote $P_1(t)$ (resp. $P_2(t)$)  the mid-price of the first stock (resp. the second stock) of the pairs trading strategy. The signal is constructed as $S_{t}:=P_1(t) - a_t \,P_2(t)+b_t$ where the coefficients $a_t$ and $b_t$ piece-wise constant, updated on each trading day by an ordinary least square regression of $P_1(\cdot)$ on $P_2(\cdot)$ over the previous $5$ days. For the KO and PEP pair, we 
regress KO (i.e., the first stock) with respect to PEP (i.e., the second stock) to construct the signal. Similarly, we treat SH as the first stock and SDS as the second stock to construct the corresponding  signal.

%As an example, we use second-by-second NYSE price records  of KO and PEP shares  during trading hours  09:30AM-4:00PM  for the period 07/01/2013 - 07/01/2020 to construct a pair-trading signal.Denote ${\rm PEP}(t)$ (resp. ${\rm CO}(t)$)  the mid-price of Pepsi (resp. Coca-Cola). The signal is constructed as $S_{t}:={\rm CO}(t) - a_t \,{\rm PEP}(t)+b_t$ where the coefficients $a_t$ and $b_t$ piece-wise constant, updated on each trading day by an ordinary least square regression of ${\rm PEP}(\cdot)$ on ${\rm CO}(\cdot)$ over the previous $5$ days. The signal  with SH and SDS is constructed in the same way.

\iffalse
\paragraph{Auto-correlation.} The auto-correlation of the signals are showed in Figure \ref{fig:autocorrelatioon}. {\color{red}[Discuss: should we keep this?]}

\begin{figure}[H]
  \centering
  \begin{subfigure}[t]{0.48\textwidth}
    \includegraphics[width=\textwidth]{figures/acf_ko_pep.png}
    \caption{\label{fig:acf_ko_pep}KO-PEP.}
  \end{subfigure}
  \begin{subfigure}[t]{0.48\textwidth}
    \includegraphics[width=\textwidth]{figures/acf_sh_sds.png}
    \caption{\label{fig:acf_sh_sds}SH-SDS.}
  \end{subfigure}
  \caption{\label{fig:autocorrelatioon}Auto-correlation for  pairs trading signals.}
\end{figure}
\fi

\paragraph{Number of level crossings.}
To assess the roughness of the signal $S$, we
analyze  the number of level crossings as a function of $\delta$ and apply Proposition \ref{prop.localtime}. Recall that $\log(D_t^{\delta}(S)) \sim -(p-1)\log(\delta)+ \text{constant}$ as $\delta\to 0$, where $p$ measures the roughness of the path. We estimate the  exponent $p$ by linear regression of $\log(D_t^{\delta}(S))$ on $\log(\delta)$. 
As shown in Figure \ref{fig:pepsi_cocacola_crossings}, the estimated exponent of the trading signal $S$ with KO and PEP is around $p=1.85$, which implies that the path is slightly smoother than the Brownian motion (for which $p=2$). On the other hand,  \ref{fig:sh_sds_crossings} shows that the estimated exponent of the trading signal $S$ with SH and SDS is around $p=1.25$ implying that $S$ is much smoother than Brownian motion.

There are two types of crossings with different time-scales along the path: crossings due to the mean reverting phenomenon on a longer time-scale and crossings with small magnitudes due to the roughness of the path once the signals revert to level $0$. The crossings of the first type could be captured by all $\delta$ with appropriate choices. Crossings of the second type show up as $\delta\rightarrow 0$. For KO-PEP signals, empirical estimates seem to indicate a non-zero limit of the realized profit $\delta D_T^{\delta}$ as $\delta\rightarrow 0$. This is consistent with the result in Proposition \ref{prop.localtime} for $p=2$, indicating that it is not more profitable to use smaller thresholds $\delta$ for trading this pair (see Figure  \ref{fig:pepsi_cocacola_realized_profit}). The realized profit is maximized at $\delta\simeq 3.2 \sigma$. %{\color{red}which is quite different from the OU model, where realized profit is maximized for $\delta< \sigma$.} 
For SH-SDS signals, empirical estimates imply a  limit of the realized profit $\delta D_T^{\delta}$ at zero as $\delta\rightarrow 0$. This is consistent with the result in Proposition \ref{prop.localtime} for $p<2$. Similar to the previous case, this  indicates that it is also not profitable to use  small thresholds $\delta$ for trading this pair (see Figure  \ref{fig:sh_sds_realized_profit}). The realized profit is maximized at $\delta\simeq 0.3 \sigma$. %, which is quite different from the OU model, where realized profit is maximized for $\delta< \sigma$.

Comparing the results in Figures \ref{fig:pepsi_cocacola_crossings} and \ref{fig:sh_sds_crossings}, we see that different pairs trading signals have very different behaviors in terms of the roughness of the path and the optimal thresholds. This is different from the results assuming Ornstein-Uhlenbeck models, where realized profit is always maximized at $\delta \lesssim \sigma$ \cite{leung2015optimal}.

\begin{figure}[H]
  \centering
  \begin{subfigure}[t]{0.48\textwidth}
    \includegraphics[width=\textwidth]{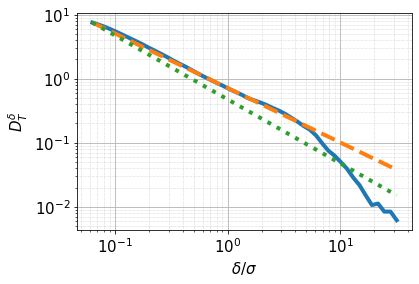}
    \caption{\label{fig:pepsi_cocacola_crossings}Number of level crossings $D_T^{\delta}$. (Orange dashed line:  $\delta^{p-1}$ with $p=1.85$; Green dotted line: asymptotics  for $p=2.0$; $T=1$ day).}
  \end{subfigure}
  \begin{subfigure}[t]{0.48\textwidth}
    \includegraphics[width=\textwidth]{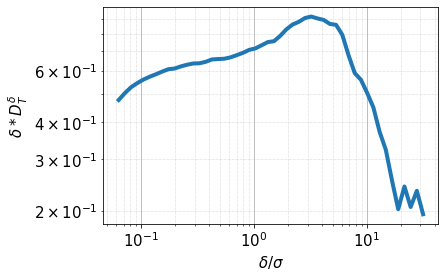}
    \caption{\label{fig:pepsi_cocacola_realized_profit}Realized profit $\delta D_T^{\delta}$ ($T=1$ day).}
  \end{subfigure}
  \caption{Number of crossings $D_T^{\delta}$ and realized profit $\delta D_T^{\delta}$ as a function of threshold $\delta$ for Coca-Pepsi pairs trading signal.}
\end{figure}

\begin{figure}[H]
  \centering
  \begin{subfigure}[t]{0.48\textwidth}
    \includegraphics[width=\textwidth]{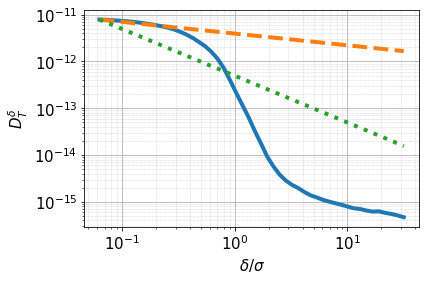}
    \caption{\label{fig:sh_sds_crossings}Number of level crossings $D_T^{\delta}$. (Orange dashed line:  $\delta^{p-1}$ with $p=1.25$; Green dotted line: asymptotics  for $p=2.0$; $T=1$ day).}
  \end{subfigure}
  \begin{subfigure}[t]{0.48\textwidth}
    \includegraphics[width=\textwidth]{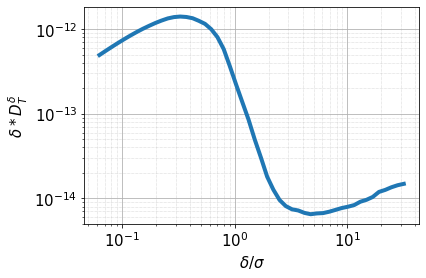}
    \caption{\label{fig:sh_sds_realized_profit}Realized profit $\delta D_T^{\delta}$ ($T=1$ day).}
  \end{subfigure}
  \caption{Number of crossings $D_T^{\delta}$ and realized profit $\delta D_T^{\delta}$ as a function of threshold $\delta$ for SH-SDS pairs trading signal.}
\end{figure}

\paragraph{Comparison with  Ornstein-Uhlenbeck models}
The most widely used model for pairs trading is the Ornstein-Uhlenbeck (OU) model \cite{gatev2006,leung2015}, mainly due to its mean-reversion properties and analytical tractability:
\begin{eqnarray}
    d S_t =  \gamma d B_t +\alpha(\mu- S_t) dt. \label{eq.OU}
\end{eqnarray}
In the stationary case, the signal has a standard deviation $\sigma= \gamma/\sqrt{2\alpha}$.

Recall that the trading strategy $\phi^+$  consists in shorting the pair when $S_t$ crosses the threshold $\delta$ from below and unwinds the position when $S_t$ returns to  $0$. Denote by $\sigma$ the (sample) standard deviation of $S_t$. In Figure~\ref{fig:pepsi_cocacola_onestd}, we provide the empirical distributions of the durations for  waiting period $(\tau_k^+-\theta_{k-1}^+)$, holding period $(\theta^+_k-\tau^+_k)$ and the maximum loss during the holding period,  when $\delta=\sigma$ is the intraday standard deviation of the signal,  a common choice for mean-reversion strategies \cite{pairs}.

As seen from the semi-logarithmic plots in Figures \ref{fig:waiting_duration} and \ref{fig:holding_duration}, the durations of the holding period and the waiting period are approximately exponentially distributed.    The maximum loss has a Pareto tail with exponent $k=1$, which is very heavy tailed and indicates infinite mean and variance, as shown by the log-log plot in Figure \ref{fig:holding_height}.
This combination of a Pareto tail for the excursion height and an exponential duration for $\delta-$excursions corresponds neither to the Brownian case nor to the case of the Ornstein-Uhlenbeck  process.

\iffalse
{\color{red}
	\begin{eqnarray}\label{eq:ou}
	d S_t =  \gamma d B_t +\alpha(\mu- S_t) dt,
	\qquad {\rm i.e,}\quad
	S_t = e^{-\alpha t} \left( S_0 + \gamma \int_0^t e^{\alpha s}dB_s\right) +\mu (1-e^{-\alpha t}),
	\end{eqnarray}
	where $\alpha \in \mathbb{R}$  and $\gamma\in \mathbb{R}$ and $B$ is a standard Brownian motion.}
\fi

For comparison, we estimate the OU model \eqref{eq.OU} using a method of moments, leading the following   parameter estimates   for the KO-PEP pair (time is measured in seconds):
\begin{eqnarray}\nonumber
\label{eq:estimated_parameters}
\widehat{\alpha} = 0.00064,\quad \widehat{\mu}=-8.634\times 10^{-5},\,\, \textit{\rm and}\,\,\widehat{\gamma} = 0.0014.
%\widehat{\alpha} = 1.3 \times 10^{-3},\quad \widehat{\mu}=-0.02,\,\, \textit{\rm and}\,\,\widehat{\gamma} = 4.8\times  10^{-3}.
\end{eqnarray}
 \begin{figure}[H]
  \centering
  \begin{subfigure}[t]{0.48\textwidth}
    \includegraphics[width=\textwidth]{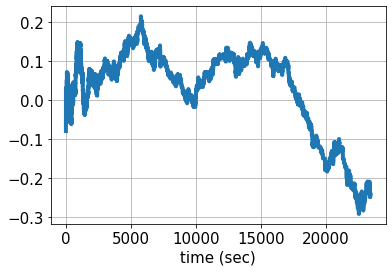}
    \caption{Coca-Pepsi pair trading signal: 07/16/2013.}
  \end{subfigure}
  \begin{subfigure}[t]{0.48\textwidth}
    \includegraphics[width=\textwidth]{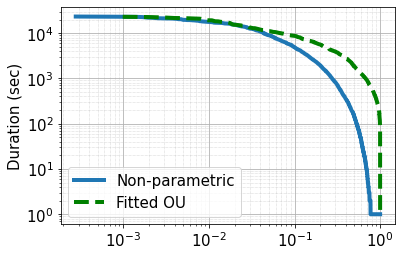}
    \caption{\label{fig:waiting_duration} Rank-frequency plot for the waiting period $(\tau_k^+-\theta_{k-1}^+)$, semi-logarithmic scale.}
  \end{subfigure}
  \begin{subfigure}[t]{0.48\textwidth}
    \includegraphics[width=\textwidth]{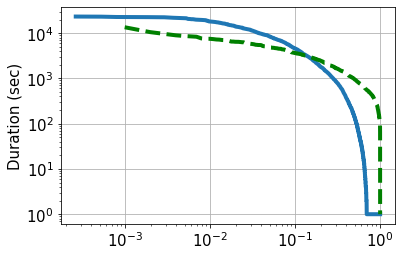}
    \caption{\label{fig:holding_duration}Rank-frequency plot for   the holding period $(\theta^+_k-\tau^+_k)$, semi-logarithmic scale.}
  \end{subfigure}
  \begin{subfigure}[t]{0.48\textwidth}
    \includegraphics[width=\textwidth]{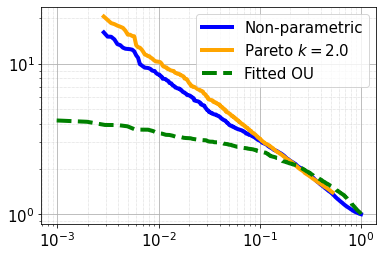}
    \caption{\label{fig:holding_height} Rank-frequency plot for maximum loss during the holding period  (log-log scale). Orange dotted line: Pareto distribution   with exponent $k={2.0}$.}
  \end{subfigure}
  \caption{\label{fig:pepsi_cocacola_onestd}KO-PEP pair trading signal ${\rm KO}(t)-a_t{\rm PEP}(t)+b_t$ at one-second frequency, for $\delta=\sigma$ (2007-2020). Data (blue) vs O-U model (dotted).}
\end{figure}
The corresponding model-based distributions for  the duration of the holding period, the waiting period and the worst loss during the holding period are displayed (green dotted lines) alongside the empirical distributions of these quantities in Figure \ref{fig:pepsi_cocacola_onestd}. The discrepancy between the green dotted lines and the blue solid lines in Figures \ref{fig:waiting_duration}, \ref{fig:holding_duration}, and \ref{fig:holding_height} illustrates that the distributions computed using the Ornstein-Uhlenbeck model give a poor approximation of the corresponding empirical distributions, leading to an inaccurate representation of the risk and return profile of the strategy. In particular, the fitted OU process underestimates duration of the holding period (see Figure \ref{fig:holding_duration}) and the maximum loss during the holding period (see Figure \ref{fig:holding_height}).
This is a strong indication of the risk of model mis-specification in such mean-reversion strategies.

For the SH-SDS pair, we have the following estimates of the Ornstein-Uhlenbeck process:
\begin{eqnarray}\nonumber
\label{eq:estimated_parameters2}
\widehat{\alpha} =0.08,\quad \widehat{\mu}=-8.634 \times 10^{-5},\,\, \textit{\rm and}\,\,\widehat{\gamma} = 0.0014.
\end{eqnarray}
 The discrepancy between the green dotted lines and the blue solid lines in Figures \ref{fig:waiting_duration2}, \ref{fig:holding_duration2}, and \ref{fig:holding_height2} illustrates that the distributions computed using the Ornstein-Uhlenbeck model give a poor approximation of all three empirical distributions including the duration of the holding period, the waiting period and the worst loss during the holding period, leading to an inaccurate representation of the risk and return profile of the strategy.  This is yet another indication of the risk of model mis-specification in such mean-reversion strategies.

\begin{figure}[H]
  \centering
  \begin{subfigure}[t]{0.48\textwidth}
    \includegraphics[width=\textwidth]{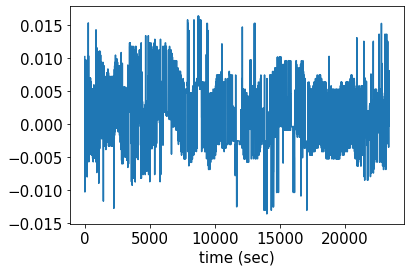}
    \caption{SH-SDS pair trading signal: 07/16/2013.}
  \end{subfigure}
  \begin{subfigure}[t]{0.48\textwidth}
    \includegraphics[width=\textwidth]{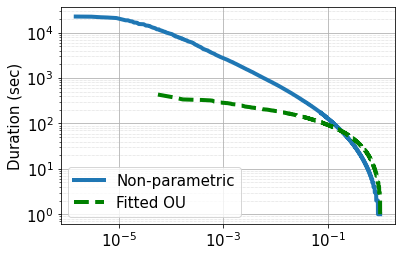}
    \caption{\label{fig:waiting_duration2} Rank-frequency plot for the waiting period $(\tau_k^+-\theta_{k-1}^+)$, semi-logarithmic scale.}
  \end{subfigure}
  \begin{subfigure}[t]{0.48\textwidth}
    \includegraphics[width=\textwidth]{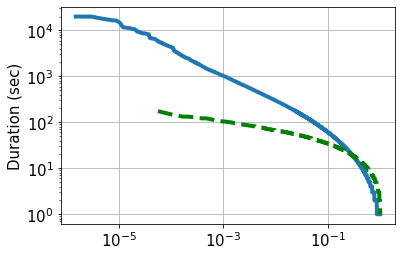}
    \caption{\label{fig:holding_duration2}Rank-frequency plot for   the holding period $(\theta^+_k-\tau^+_k)$, semi-logarithmic scale.}
  \end{subfigure}
  \begin{subfigure}[t]{0.48\textwidth}
    \includegraphics[width=\textwidth]{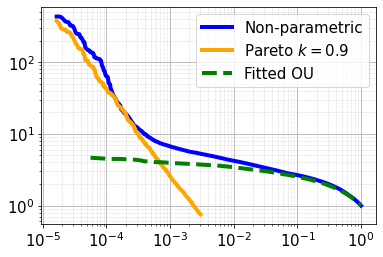}
    \caption{\label{fig:holding_height2} Rank-frequency plot for maximum loss during the holding period  (log-log scale). Orange dotted line: Pareto distribution with exponent $k={0.9}$.}
  \end{subfigure}
  \caption{\label{fig:pepsi_cocacola_onestd2}SH-SDS pair trading signal ${\rm SH}(t)-a_t{\rm SDS}(t)+b_t$ at one-second frequency, for $\delta=\sigma$ (2007-2020). Data (blue) vs O-U model (dotted).}
\end{figure}

It is worth pointing out that the distribution of the rank-frequency plot has two regimes (Figure \ref{fig:holding_height2}):  excursions of small amplitude, associated with  small losses, behave as in the Ornstein-Uhlenbeck  model  whereas excursions of large amplitudes, associated  with  large losses, exhibit a heavy tail with an amplitude which follows a Pareto distribution with exponent $k=0.9$, which resembles the case of Brownian excursions \cite{ananova2021}.
The excursions of these trading signals thus seem to interpolate between the OU model for small amplitudes and the (driftless) Brownian case for large amplitudes, as   observed in \cite{ananova2021}. 

 \section{Model-free scenario simulation}\label{sec:modelfreesimulation}
The above examples  illustrate that, to correctly reflect the risk and return of dynamic trading strategies, a model needs to adequately reflect the   excursion properties of trading signals.
We now show how
Proposition
\ref{prop.decomposition} may be used to design a {\it non-parametric} method for simulating paths whose excursion properties match those observed in data.
The idea is to
resample from the set of empirical excursions, and concatenate them to generate paths.
We will now discuss this idea in some detail.

Given an observed (continuous) path $S$,   we  may decompose  $S$ as in Proposition
\ref{prop.decomposition} into a sequence $(e_k)_{k= 1, \ldots,D^{\delta}_T}$ of $\delta-$excursions  defined as in \eqref{eq.dexcursions}.
%$e_k$ take values in ${\cal U}_\delta$.
Denoting by $\epsilon_x$ a unit point mass at $x$, we define the {\it empirical $\delta$-excursion measure} 
    $$ \Pi_T(S)= \frac{1}{D^{\delta}_T(S)}\ \sum_{k=1}^{D^{\delta}_T(S)} \epsilon_{e_k}. $$
$\Pi_T(S)$    is a probability measure on the set ${\cal U}_\delta$ of $\delta-$excursions. One may sample from $\Pi_T$  by randomly resampling from the empirical sequence of excursions $(e_k)$.

From Ito's theory of excursions \cite{ito1972} we know that, for a Markov process, $\delta-$excursions are IID and  $\Pi_T(S)$  is thus the empirical distribution associated to a sample of size $D^{\delta}_T(S)$. If furthermore $S$ is recurrent at zero, $D^{\delta}_T(S)\to \infty$ as $T\to\infty$ so one may recover the distribution $\Pi^\delta$ of $\delta-$excursions and reconstruct the (law of) $S$ as an IID concatenation of $\delta-$excursions. 

{We now give conditions under which the structure of $S$ may be {\it recovered} from its $\delta-$excursions, providing we observe sufficiently many of them.} 
 
Let $(\mathcal{F}^S_t)_{t\geq 0}$ be the natural filtration of $S$ and define the {\it shift} operator 
\begin{eqnarray}
     \Theta_t : \mathcal{E}& \mapsto & \mathcal{E}\nonumber\\
   f  &  \rightarrow  &\Theta_t(f):=f(\cdot+t)
\end{eqnarray}
\begin{assumption}\label{ass.regenerative} \ \\
\begin{enumerate}  
    \item[(i)] $S$ is recurrent at $0$: \begin{eqnarray}\label{eq:recurrent}
\mathbb{P}(T_0^0(S)<\infty) = 1.
\end{eqnarray}
\item[(ii)]  $S$ is regenerated at zero:
there exists a measure $\mathbf{P}_0$  {on $\mathcal{E}$ such that for any $(\mathcal{F}_t)_{t\geq 0}$-stopping time $\tau$}, 
\begin{eqnarray}\label{eq:regenerative}
 \mathbb{P}(\Theta_{\tau}(S)\in \cdot \vert \mathcal{F}_{\tau}) = \mathbf{P}_0,\,\, \mathbb{P}-\text{almost surely on } \{\tau<\infty, S_{\tau}=0\}.
\end{eqnarray}
\end{enumerate}
\end{assumption}
Whereas a strong Markov process is regenerated at {\it any} stopping time, we only require this property at zero crossing times in (ii).
Construction of regenerative processes by concatenation of independent excursions has been studied by Lambert and Simatos \cite{lambert2014} and Yano \cite{yano2015}. Our $\delta-$excursion concept is related to, but slightly different from, the concept of 'big' excursion in \cite{lambert2014}; elements of $\Gamma_\delta$ may be seen as `large excursions' in the sense of \cite{lambert2014} and relate to $\delta-$excursions through the last exit decomposition, as noted in Lemma \ref{lemma.Ud}.

Under Assumption \ref{ass.regenerative},
the $\delta-$excursions $e_k$ are IID  variables with values in ${\cal U}_\delta$, whose law we denote $\Pi^\delta$. Property (i) then implies
     $$ D^\delta_T(S)\mathop{\to}^{T\to\infty} \infty. $$
   We can thus apply the law of large numbers to the empirical excursion measure: $\Pi_T$ approximates $\Pi^\delta$ for large $T$ and for any {Glivenko-Cantelli class ${\cal G}$ of functions \cite{talagrand1987} on ${\cal U}_\delta$}, representing properties of $\delta$-excursions, we have 
 \begin{eqnarray}\label{eq:empirical_measure_conv}
 \forall H\in {\cal G} \text{ and } f\in \mathcal{U}_\delta,\qquad  \int H(f) \Pi_T(df)\mathop{\to}^{T\to\infty} \int H(f) \Pi^\delta(df).
 \end{eqnarray}

 %The shortcomings of   model-based methods may be addressed by using a  {\it non-parametric} scenario simulation approach to evaluate the quantities of interest as described in Table \ref{table:nonparametric}: we decompose the empirical path into $\delta-$excursions and generate scenarios by random concatenation of such empirical $\delta-$excursions.
 
 This leads to a {\it non-parametric} approach for {  scenario simulation} based on $\delta-$excursions (see Table \ref{table:nonparametric}). {We first decompose the signal $(S_t, t\in [0,T])$ into $N:=D^\delta_T(S)$  $\delta$-excursions, then we construct a random sequence of such excursions by uniformly sampling this set with replacement and generate a new path by concatenating this random sequence of $\delta$-excursions. 
 
 } 
 
 The paths generated in this way have $\delta-$excursions whose properties mimic those of $S$, without requiring prior knowledge about $S$. By construction this leads to a regenerative process recurrent at zero, thus satisfying Assumption \ref{ass.regenerative}, and as $T\to \infty$ we recover the law of the data generating process $S$.

\begin{table}[H]
    \centering
    \begin{tcolorbox}[colback=red!5!white,colframe=red!75!black,fonttitle=\bfseries,title=Non-parametric scenario simulation by pasting of excursions]
\textbf{Input data} : sample path $(S_t, t\in [0,T])$
\begin{enumerate}
\item   Decompose $(S_t, t\in [0,T])$ into $\delta-$excursions $e_1, \ldots, e_N\in {\cal U}_\delta$ using Proposition \ref{prop.pathdecomposition}.
     \item Generate an IID sequence of integers $(k_1(\omega), k_2(\omega),  \ldots )$ where $$k_i\sim  {\rm UNIF}(\{1,2, \ldots, N\}).$$
     \item Construct a path $X(\omega)$ as in \eqref{eq.dexcursions} by concatenating the excursions in the order given by $(k_i,i\geq 1)$:
     $$ X_t(\omega)= \sum_{i\geq 1} e_{k_i(\omega)}(t-\theta^+_{i-1})\qquad{\rm where}\qquad \theta^+_0=0,\quad \theta^+_i= \sum_{j=1}^{i} \Lambda(e_{k_i}).$$
 \end{enumerate}
 \textbf{Output} : simulated sample path $X$
\end{tcolorbox}
\caption{Non-parametric scenario simulation by pasting of excursions.}
    \label{table:nonparametric}
\end{table}

{Figure \ref{fig:pasting} shows  sample paths generated using the non-parametric scenario simulation method with the KO-PEP signals during 2007-2020; see the details of the signal construction in Section \ref{sec:behavior}. We construct sample paths with four randomly sampled $\delta$-excursions (see Figure \ref{fig:pasting} {\bf Left} for the case of $\delta=\sigma$ and Figure \ref{fig:pasting} {\bf Right} for the case of $\delta=2\sigma$). By construction, paths generated in this way retain the roughness properties of the observed path as well as the empirical distribution  of  heights and  durations of $\delta-$excursions.}

More generally one may relax the regenerative assumption (ii) and consider ergodic dynamical systems in a more general setting such as stochastic dynamical systems driven by fractional Brownian motons \cite{hairer2005}, but this is beyond the scope of the present work and is left for future work.

\begin{figure}[H]
    \centering
\includegraphics[width=0.45\textwidth]{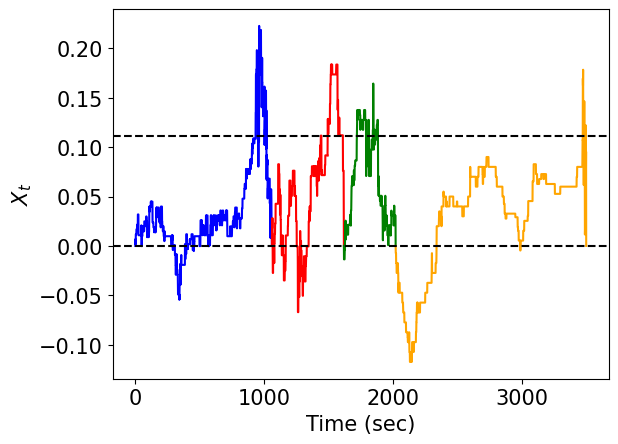}
\includegraphics[width=0.45\textwidth]{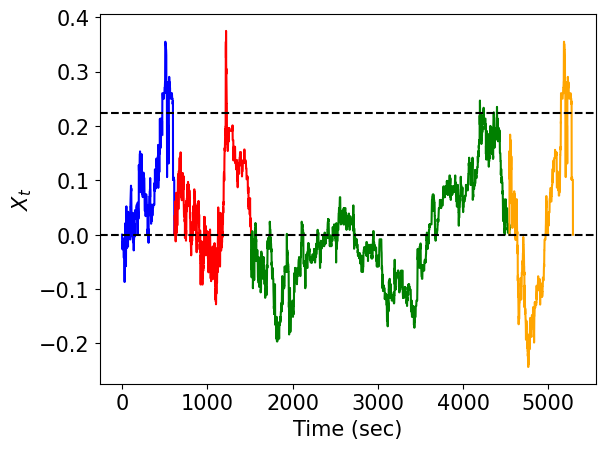}
    \caption{Example of sample path generated via the non-parametric method, using KO-PEP signals ({\bf Left}: $\sigma=\delta$; {\bf Right}: $\sigma=2\delta$).}
    \label{fig:pasting}
\end{figure}

\bibliographystyle{siam}
\bibliography{excursion.bib}
\newpage

\end{document}